\documentclass[aps,10pt,twocolumn,pra,showpacs,superscriptaddress]{revtex4-1}

\newcounter{intro}

\usepackage[sc]{mathpazo}
\usepackage{amsmath}
\usepackage{amssymb}
\usepackage{amsthm}
\usepackage{lmodern}
\usepackage{thmtools,thm-restate}

\usepackage{enumitem}
\usepackage{color}

\usepackage{stmaryrd}

\DeclareFontFamily{U}{mathx}{\hyphenchar\font45}
\DeclareFontShape{U}{mathx}{m}{n}{
      <5> <6> <7> <8> <9> <10>
      <10.95> <12> <14.4> <17.28> <20.74> <24.88>
      mathx10
      }{}
\DeclareSymbolFont{mathx}{U}{mathx}{m}{n}
\DeclareFontSubstitution{U}{mathx}{m}{n}
\DeclareMathAccent{\widecheck}{0}{mathx}{"71}
\DeclareMathAccent{\wideparen}{0}{mathx}{"75}

\RequirePackage{filecontents}        

\usepackage{natbib}

\usepackage{algorithm}
\usepackage{algpseudocode}

\usepackage[margin=1in]{geometry}
\def\complex{\mathbb{C}}

\def\natural{\mathbb{N}}
\def\integer{\mathbb{Z}}
\def\field{\mathbb{F}}

\def\conj{^{\dagger}}
\def\to{\rightarrow}

\newcommand{\1}{\mathbb{1}}

\newcommand{\rowspan}[1]{\operatorname{span}\{#1\}}

\def \lket {\left|}
\def \rket {\right\rangle}
\def \lbra {\left\langle}
\def \rbra {\right|}
\newcommand{\ket}[1]{\lket\mspace{0.5mu} #1 \mspace{0.5mu}\rket}
\newcommand{\encket}[1]{\ket{\overline{#1}}}
\newcommand{\bra}[1]{\lbra\mspace{0.5mu} #1 \mspace{0.5mu}\rbra}

\newcommand{\kb}[1]{\ket{#1}\bra{#1}}
\def\logX{\overline{X}}
\def\logZ{\overline{Z}}
\def\dsl{\llbracket}
\def\dsr{\rrbracket}

\def\css{\sansserif{CSS}}

\def\rs{\sansserif{RS}}
\def\prs{\sansserif{PRS}}
\def\srs{\sansserif{SRS}}

\def\A{\mathcal{A}}
\def\B{\mathcal{B}}
\def\C{\mathcal{C}}

\def\G{\mathcal{G}}
\def\H{\mathcal{H}}

\def\K{\mathcal{K}}

\def\M{\mathcal{M}}

\def\P{\mathcal{P}}
\def\Q{\mathcal{Q}}
\def\R{\mathcal{R}}

\def\V{\mathcal{V}}

\def\g{\operatorname{G}}

\def\h{\operatorname{H}}

\def\tC{\widetilde{\C}}

\def\th{\widetilde{H}}

\newtheorem{theorem}{Theorem}
\newtheorem*{theorem*}{Theorem}
\newtheorem{lemma}[theorem]{Lemma}

\newtheorem{definition}[theorem]{Definition}

\def\path{\operatorname{P}}

\newcommand{\sansserif}[1]{%
  \ifmmode
  \mathsf{#1}%
  \else
   \textsf{#1}%
  \fi
}

\def\ev{\sansserif{ev}}

\renewcommand{\int}[2]{\Delta_{#1} \G'_{#2}}

\usepackage[toc,page]{appendix}

\begin{document}
\title{Towards low overhead magic state distillation}
\author{Anirudh Krishna}
\affiliation{
  Universit\'e de Sherbrooke,
  2500 Boulevard de l'Universit\'e,
  Sherbrooke, QC J1K 2R1, Canada
  }
\author{Jean-Pierre Tillich}
\affiliation{
  Inria, Team SECRET,
  2 Rue Simone IFF, CS 42112,
  75589 Paris Cedex 12, France,
  }
\begin{abstract}
Magic state distillation is a resource intensive sub-routine for quantum computation.
The ratio of noisy input states to output states with error rate at most $\epsilon$ scales as $O(\log^{\gamma}(1/\epsilon))$ \cite{bravyi2012magic}.
In a breakthrough paper, Hastings and Haah \cite{hastings2017distillation} showed that it is possible to construct distillation routines with sub-logarithmic overhead, achieving $\gamma \approx 0.6779$ and falsifying a conjecture that $\gamma$ is lower bounded by $1$.
They then ask whether $\gamma$ can be made arbitrarily close to $0$.
We answer this question in the affirmative for magic state distillation routines using qudits of prime dimension ($d$ dimensional quantum systems for prime $d$).
\end{abstract}
\date{\small \today}

\maketitle

\emph{Introduction -- .}
One of the biggest obstacles we face as we endeavour to construct a fault-tolerant quantum computer is the tremendous resource overhead required to implement a universal set of gates.
Analyses of the resource usage in quantum circuits point to \emph{magic state distillation} as the biggest bottleneck \cite{suchara2013comparing,jones2013multilevel,o2017quantum}.
This sub-routine consumes several impure copies of a resource state to produce fewer, but purer, copies of this state.
The figure of merit used to measure the efficiency of this protocol is called the \emph{overhead}, defined as the ratio of the number of input to output states.
For a target error rate of $\epsilon$, the overhead scales as $O\left(\log^{\gamma}\left(\frac{1}{\epsilon}\right)\right)$ where $\gamma$ depends on the details of how the distillation sub-routine is constructed.
It was conjectured that $\gamma \geq 1$, which if true would imply a strict bound on the efficiency of magic state distillation \cite{bravyi2012magic}.

In a breakthrough article \cite{hastings2017distillation}, Hastings and Haah have recently constructed codes for which $\gamma \approx 0.6779$ is achievable and thereby falsified the above conjecture.
Their proof uses a quantum error correcting code with a very large block size (the number of qubits that will be addressed at one time step); they require a block size of roughly $2\times 10^{17}$.
This raises the question of whether $\gamma$ can be reduced arbitrarily close to zero and whether this is possible with a smaller block size.
In this paper, we answer these questions in the affirmative.

Our result employs quantum error correcting codes defined over qudits or $d$ dimensional quantum systems.
To the best of our knowledge, earlier qudit magic state distallation protocols such as \cite{anwar2012qutrit,campbell2012magic,howard2012qudit,campbell2014enhanced} are only capable of distilling a single qudit.
We construct quantum error correcting codes which can encode a growing number of logical qudits which also possess special symmetry properties desirable for fault tolerance.
In particular, we build on the classification of diagonal gates in the Clifford hierarchy by Cui et al. \cite{cui2017diagonal} to generalize the framework of tri-orthogonality \cite{bravyi2012magic} to qudits of dimension $p$, where $p$ is some prime.
We then use Reed-Solomon codes over a prime field to show that we can achieve $\gamma = O(1/\log(p))$.

\emph{Gates from the third level of the Clifford hierarchy -- .}
For a prime $p$, the state space of a qudit of dimension $p$ is associated with the complex Euclidean space $\complex^{p}$.
The generalized Pauli group over $\complex^{p}$ is defined using the shift and boost operators $X$ and $Z$ where for $j \in \field_{p}$,
\begin{align}
  X\ket{j} = \ket{j + 1}, \quad Z\ket{j} = \omega^{j} \ket{j}~,
\end{align}
where addition is modulo $p$ and $\omega = \exp\left( 2\pi i /p \right)$ is the $p$-th root of unity.
As shorthand, we shall let $X^{g} = \bigotimes_{i} X^{g_i}$ and $Z^{f} = \bigotimes_{j} Z^{f_j}$ for $g,f \in \field_{p}^{n}$.
For $p>2$, the Pauli group is the group generated by $X$ and $Z$ along with a phase
\begin{align}
  \P = \langle \omega\1, X, Z \rangle~.
\end{align}
The $n$-qudit Pauli group $\P_n = \P^{\otimes n}$ is then the $n$-fold tensor product of $\P$.
aoeuidhtnspyfgcrlqjkxbmw
Gottesman and Chuang introduced the Clifford hierarchy \cite{gottesman1999demonstrating}, which is defined recursively as
\begin{align}
  \K^{(t)} = \{U | UPU\conj \in \K^{(t-1)}, \forall P \in \P\}~.
\end{align}
The second level of this hierarchy is the automorphism group of the Pauli group, denoted $\K^{(2)}$, and is called the Clifford group.
Although Clifford gates form a closed finite group, we can use them together with any gate in the third level of the hierarchy to compose a universal set of gates for quantum computation \cite{campbell2012magic}.
We shall restrict our attention to diagonal gates in the Clifford hierarchy, and use $\K^{(t)}_d$ to denote the diagonal gates in the $t$-th level of the hierarchy.

For $m \in \natural$ and $1 \leq a \leq p-1$, we let $U_{m,a}$ be the single-qudit gate
\begin{align}
  U_{m,a} = \sum_{j} \exp\left(\frac{2\pi i}{p^m} j^a \right)\kb{j}~.
\end{align}
These gates are characterized by the parameter $m$ known as the \emph{precision} of the unitary gate and the degree $a$ of the monomials $j^{a}$ in the exponent of the phase.

The following result was shown by Cui et al. \cite{cui2017diagonal} (see theorem 2).
\begin{lemma}[Cui et al.]
  \label{lem:umaContainment}
  For $m\in \natural$, $1 \leq a \leq p-1$,
    \begin{align}
      U_{m,a} \in \K^{((p-1)(m-1) + a)}_d \setminus \K^{((p-1)(m-1) + a-1)}_d~.
    \end{align}
\end{lemma}

If the dimension of the qudit is $p=3$ i.e. a qutrit, the third-level gate is $U_{2,1}$ whereas for qudits of dimension $p>3$, the third-level gate is $U_{1,3}$.
For the main body of the article, we shall restrict our attention to $p \geq 5$.
We have dealt with the special case of $p=3$ as well in appendix A.

Before proceeding to the next section, we recall that an $n$-qudit unitary gate $U$ is \emph{transversal} if it can be expressed as $\otimes_{i=1}^n V_i$, where $V_i$ is some single-qudit unitary operator \cite{gottesman1997stabilizer}.
Since these gates act on each qudit independently, such gates prevent the spread of errors on one qudit from spreading to another and are therefore fault-tolerant by construction.
We seek quantum error correcting codes such that logical gates from the third-level of the Clifford hierarchy can be implemented transversally and this motivates the framework of tri-orthogonality.

\emph{Generalized tri-orthogonality--.}
In this section, we extend the definition of a triply-even space and a tri-orthogonal code as given in \cite{bravyi2012magic}, \cite{haah2017codes}.

\begin{definition}[\textbf{Tri-orthogonal matrix}]
	\label{def:triorthogonalMatrix}
	Let $\h \in \field_{p}^{m \times n}$ be a matrix whose rows are labeled $\{h^{a}\}_{a=1}^{m}$.
    We say that $\h$ is \emph{tri-orthogonal} if
	\begin{enumerate}
    \item \label{item:biorthog} for $1 \leq a < b \leq m$,
    \[
      \sum_{i} h^{a}_i h^{b}_i = 0 \pmod{p}~.
    \]
		\item \label{item:triorthog} for $1 \leq a < b < c \leq m$,
		\[
			\sum_{i} h^{a}_i h^{b}_i h^{c}_i = 0 \pmod{p}~.
		\]
	\end{enumerate}
\end{definition}

We shall partition $\h$ into two matrices $\h_0$ and $\h_1$, where $\h_0$ contains the $(m-k)$ rows of $\h$ such that for $h \in \h_0$, its weight $\sum_i (h_i)^2 = 0 \pmod{p}$ and $\h_1$ contains the $k$ rows of $\h$ such that for $h \in \h_1$, its weight $\sum_j (h_j)^2 \neq 0 \pmod{p}$ for some natural number $k \leq M$.
Let $\g$ be the matrix whose rows are orthogonal to $\h$.
Denote by $\H_0,\H_1,\H$ and $\G$ the span of the rows of the matrices $\h_0$, $\h_1$, $\h$ and $\g$ respectively.
To obtain a tri-orthogonal quantum code $\css(X,\h_0; Z,\g)$ from the tri-orthogonal matrix $\h$, we associate 
\begin{enumerate}
  \item the rows of $\h_0$ with the $X$ stabilizer generators;
  \item the rows of $\g$ with the $Z$ stabilizer generators;
  \item the rows of $\h_1$ to both the $X$ and $Z$ logical operators.
\end{enumerate}

Note that when we work with qubits, assuming that the weight of a vector $\sum_i (h_i)^2 = 1 \pmod{2}$ implies that $\sum_i (h_i)^t = 1 \pmod{2}$ for all $t \in \natural$.
However this is not the case when working with other primes.
For instance, consider the vector $u := (0,1,2) \in \field_3$.
It obeys $\sum_i (u_i)^3 = 0 \pmod{3}$ but $\sum_i (u_i)^2 = 2 \pmod{3}$.
For this reason we define the quantity $\epsilon_a := \sum_i (h^{a}_i)^3$ for $1 \leq a \leq k$.

\begin{restatable*}{lemma}{validCode}
\label{lem:validCode}
  The code $\css(X,\h_0; Z,\g)$ is a valid quantum code with dimension $k$ and distance $d$, where
  \begin{align*}
    d = \min_{v \in \H_1 \setminus \G}(v)~.
  \end{align*}
\end{restatable*}
The proof of this statement has been relegated to appendix B.

Since $\css(X,\h_0;Z,\g)$ is a CSS code, for any $u \in \field_p^k$, the encoded state $\encket{u}$ can be expressed as
\begin{align*}
	\encket{u} = \frac{1}{\sqrt{|\H_0|}} \sum_{h \in \H_0} \ket{\sum_{a=1}^{k} u_a h^{a} + h}~,
\end{align*}
where the arithmetic within the ket is performed mod $p$.

\begin{theorem}
	\label{thm:genBraHaah}
	Let $\h \in \field_p^{m \times n}$ be a tri-orthogonal matrix whose rows are labelled $\{h^a\}_{a=1}^{m}$.
	$\C(\h) = \css(X,\h_0; Z, \g)$ be the tri-orthogonal code obtained from $\h$.
	For $1 \leq a \leq k$, $\epsilon_a := \sum_i (h_i^a)^3$.
  The transversal physical gate $U_{1,3}^{\otimes n}$ performs the following transversal logical gate:
	\begin{align*}
		 U_{1,3}^{\otimes n} \encket{u} = \bigotimes_{a=1}^{k} \left(\overline{U_{1,3}}\right)^{\epsilon_a} \encket{u}~.
	\end{align*}
\end{theorem}
\begin{proof}
	Let $f \in \rowspan{h^a}_{a=1}^{m} \subseteq \field_{p}^{n}$ be a vector in $\H$.
	The action of transversal $U_{1,3}$ on $\ket{f}$ can be expressed as
	\begin{align}
		U_{1,3}^{\otimes n}\ket{f} &= \prod_{i=1}^{n}\exp\left(\frac{2\pi i}{p}f_i^{3} \right)\ket{f}\\
								   &= \exp\left( \frac{2\pi i}{p} \sum_{i=1}^{n}f_i^{3} \right)\ket{f}~.
	\end{align}
	By assumption, we may express $f = \sum_{a} u_{a} h^{a}$ for some constants $\{u_a\}_a \in \field_p$ and therefore for $i \in \{1,..,N\}$, we have
	\begin{align}
		f_i^3 &= \left(\sum_{a} u_a h^{a}_{i}\right)^3\\
			  &= \sum_{a} u_a^3 (h^{a}_{i})^3
			       + \\
             &3\sum_{a<b} \left[(u_a)^2 u_b (h^{a}_{i})^2 h^{b}_i + u_a (u_b)^2 h^{a}_{i} (h^{b}_i)^2\right]\nonumber\\
			       + &\sum_{a<b<c} u_a u_b u_c h^{a}_{i}h^{b}_{i} h^{c}_{i}\nonumber
	\end{align}
	We may use the definition of tri-orthogonality (definition \ref{def:triorthogonalMatrix}) and $\epsilon$ to simplify this expression and obtain
	\begin{align}
		\sum_{i=1}^{n} f_i^3 &= \sum_{a=1}^{k} u_{a}^{3} \epsilon_a~.
	\end{align}
	Hence it follows that
	\begin{align}
		U_{1,3}^{\otimes n}\ket{f} &= \exp\left( \frac{2\pi i}{p} \sum_{a=1}^{k} u_{a}^{3}\epsilon_a \right)\ket{f}~,
	\end{align}
	which is the desired result.
\end{proof}

Triorthogonality can also be expressed simply using the $*$-product between two vectors $u,v$ which is denoted $u*v$ and represents the vector $(u_iv_i)_i$.
This can naturally be extended to the $*$-product between $3$ or more vectors.

\begin{definition}[\textbf{Triply even space}]
  A subspace $\V \subseteq \field^{N}_p$ is said to be \emph{triply-even} if for any triple $u,v,w \in \V$,
  \begin{align*}
    |u * v * w| = 0 \pmod{p}~.
  \end{align*}
\end{definition}

For our purposes, it will be convenient to express the triply-even property slightly differently as
\begin{lemma}
	\label{lem:triplyEvenEquivalance}
	$\V$ is triply-even if and only if $\V * \V \subseteq \V^{\perp}$.
\end{lemma}
\begin{proof}
	If $\V$ is triply-even, then it means that for $u,v \in \V$, $u*v \in \V^{\perp}$.
	The other direction follows trivially.
\end{proof}

Let $\tC = [n+k,n+k-m]$ be a linear code over $\field_p$ defined by a parity check matrix $\th$ over $\field_p$ of dimension $m \times (n+k)$.
If the dual code $\tC^{\perp}$ is triply-even and contains the all $1$s vector, we may construct a quantum code by puncturing $\tC^{\perp}$.
For simplicity, suppose we puncture the first $k$ positions of the space; if not, we can always permute the qudits.
We may then express the parity check matrix $\th$ of the code $\tC$ in (almost) systematic form with respect to the puncture:
\begin{align}
  \th =
  \begin{pmatrix}
    -\1_k & \h_{1}\\
    0  & \h_{0}
  \end{pmatrix} ~,
\end{align}
where $\1_k$ represents the $k \times k$ identity matrix, and $\h_{1}$, $\h_{0}$ are matrices of dimension $k \times n$ and $(m-k) \times n$ respectively.
Note that this is not exactly systematic form because there is a minus sign preceding the identity $\1_k$.

\begin{restatable}{lemma}{triplyEvenToTriorthogonal}
\label{lem:triplyEvenToTriorthogonal}
  Let $\tC$ be an $[n+k,n+k-m]$ code such that $\tC^{\perp}$ is triply-even and contains the all $1$s vector.
  The parity check matrix $\h$ of the code $\C$
\begin{align}
  \h = 
  \begin{pmatrix}
    \h_{1}\\
    \h_{0}
  \end{pmatrix}
\end{align}
is tri-orthogonal.
Furthermore, for all $1\leq a \leq k$, $h^a \in \h_{1}$, we must have $\epsilon_a = \sum_i (h^a_i)^3 = 1 \pmod{p}$ and $h^a \in \h_{1}$ that $\sum_i (h^a_i)^2 = -1 \pmod{p}$.
\end{restatable}
The proof of this claim can be found in the appendix \ref{subsec:triplyEvenToTriorthogonal}.

The protocol for performing distillation has been presented in the paper by Bravyi and Haah \cite{bravyi2012magic}.
Although this protocol and analysis was originally meant for qubits, it readily extends to qudit systems as well.
For a target error rate of $\epsilon$, the overhead for a $\dsl n,k,d\dsr$ quantum error correcting code therefore scales as $\log^{\gamma}(1/\epsilon)$, where $\gamma = \log(n/k)/\log(d)$.

\emph{Reed-Solomon codes --.} As noted, the condition of tri-orthogonality can be cast in terms of the star-product.
Incidentally, this is a well-studied object in classical coding theory \cite{randriambololona2015products}.
Based on this, we expect that satisfying the tri-orthogonality condition with a good minimum distance becomes easier as the size of the field over which the codes are defined increases.

Reed-Solomon codes are arguably the simplest class of algebraic codes that illustrate this idea.
Let $\R$ be the ring of univariate polynomials over the field $\field_{p}$, i.e. $\R = \field_p[x]/(x^p-x)$.
We define the evaluation map $\ev:\R \to \field_p$ for any polynomial as $\ev(\alpha) = (\alpha(u))_{u \in \field_p}$.
Note that the map $\ev$ is a bijection and in particular, maps the $*$-product to the product of monomials, i.e. if $\alpha, \beta \in \M$, then
\begin{align}
	\ev(\alpha)*\ev(\beta) = \ev(\alpha\beta)~.
\end{align}
We are now equipped to define the Reed-Solomon code.
\begin{definition}
	\label{def:ReedSolomon}
	For $l \leq p$, the Reed-Solomon code $\rs_l$ is defined as
	\begin{align*}
		\left.\right.\rs_l := \operatorname{span}\{ \ev(\alpha) | \alpha \in \R, \deg{\alpha} < l\}~.
	\end{align*}
\end{definition}

We state some well known properties of Reed-Solomon codes and point the interested reader to \cite{macwilliams1977theory} for proofs.
\begin{lemma}
\label{lem:dualRS}
	For $m \leq p$, the Reed-Solomon code $\rs_l$ is a $[p,l,p-l+1]$ code over $\field_p$.
  Its dual $\rs_l^{\perp}$ is
	\[
		\rs_l^{\perp} = \rs_{p-l}~.
	\]
\end{lemma}

These observations imply the following simple condition to state when a Reed-Solomon code $\rs_l$ is triply-even.
\begin{theorem}
	The Reed-Solomon code $\rs_l$ is triply even if $3l < p+1$.
\end{theorem}
\begin{proof}
	According to lemma \ref{lem:triplyEvenEquivalance}, the condition for a Reed-Solomon code to be triply even is
	\begin{align*}
		\rs_l * \rs_l \subseteq \rs_l^{\perp} = \rs_{p-l+1}~,
	\end{align*}
	where the last equality follows from lemma \ref{lem:dualRS}.
	Since the map $\ev$ is a bijection,
	\[
		\ev(\alpha)*\ev(\beta)= \ev(\alpha\beta)~.
	\]
	For $\alpha,\beta \in \rs_{l}$, $\deg(\alpha\beta) < 2l -1$.
	This implies that the code $\rs_l$ is triply even if $2l-1 < p-l$ or equivalently if $3l < p+1$.
\end{proof}

For $u \in \field_p^{p}$, and any set $\A \subseteq \integer_p$, let $u|_{\A}$ denote the restriction of the vector $u$ to the indices in $\A$.
In particular, if $\A = \{a_i\}_{i=1}^{k} \subseteq \integer_p$ is some set of punctured locations, then
we can define a shortened Reed-Solomon code $\srs_{p-l,\A}$ defined as
\begin{align}
  \srs_{p-l,\A} = \{u|_{\A^c} | u \in \rs_{p-l}, u|_{\A} = 0\}~,
\end{align}
where $\A^c$ denotes the complement of $\A$ within $\integer_p$.

Let $\srs_{p-l,\A}^{\perp} = \prs_{l,\A}$ denote the punctured Reed-Solomon code defined as
\begin{align}
  \prs_{l,\A} = \{u|_{\A^c} | u \in \rs_{l} \}~.
\end{align}

In constructing a tri-orthogonal quantum code by puncturing $\rs_{l}$, the $X$ stabilizers correspond to the set $\srs_{l,\A}$ and the $Z$ stabilizers correspond to the set $\srs_{p-l,\A}$.
Therefore the distance of the quantum code is the distance of $\srs_{l,\A}^{\perp} = \prs_{p-l,\A}$.

\begin{lemma}
  The punctured Reed-Solomon code $\prs_{p-l,\A}$ has distance $l-k$.
\end{lemma}
\begin{proof}
	$\prs_{p-l,\A}$ contains polynomials from $\rs_{p-l}$ with degree strictly less than $p-l$.
	If we puncture the code at locations in $\A$, the distance can drop by at most $k$.
	Let $\A^c = \integer_p \setminus \A$ and let $\B \subset \A^c$ be an arbitrary subset such that $|\B| = k$.
	This implies that the polynomial $g:= \prod_{b \in \B} (x-b)$ is non-zero on the locations of the punctures and therefore, that $\ev(g)|_{\A^c}$ has weight $l-k$.
\end{proof}

Together with lemma \ref{lem:triplyEvenToTriorthogonal}, these results imply the following lemma:
\begin{lemma}
	For prime $p$, choose natural numbers $l,k$ such that $3l \leq p+1$ and $k \leq l$.
	The tri-orthogonal quantum code $\Q$ whose $X$ stabilizers are $\srs_{l,\A}$ and $Z$ stabilizers are $\srs_{p-l,\A}$ is a $\dsl p-k,k,l-k\dsr$ quantum code.
\end{lemma}

The overhead associated with $\Q$ is then just
\begin{align}
  \gamma &= \frac{\log((p-k)/k)}{\log(m-k)}~.
\end{align}
We find that $\gamma = 0.98 < 1$ for a $\dsl 35,6,6 \dsr$ code over $p=41$ and $\gamma = 0.657..$ for a $\dsl 83,14,15\dsr$ code for $p=97$.
Whereas Hastings and Haah required block sizes of roughly $2 \times 10^{17}$ to achieve sub-logarithmic overhead, we find that it is possible to exceed the $\gamma$ they could achieve using a block size of $83$.
We hasten to add this comes at the cost of using very large qudits whose dimension scales linearly with the block size.





\emph{Conclusion--.}
We have demonstrated that we can achieve arbitrarily close to constant overhead for magic state distillation using prime dimensional qudits.
Furthermore, we have done so with relatively small block sizes, although the dimension of the qudits is fairly large.
It still remains an open question to show that $\gamma$ can be made arbitrarily close to $0$ for qubit based protocols.
Moreover, it is important to know whether this can be done for small block sizes.

\emph{Acknowledgements--.}
We would like to thank David Poulin for many helpful discussions and Colin Trout for detailed feedback on an earlier draft of this work.
A.K. would like to thank Inria, Paris for their hospitality during his visit and the MITACS organization for the Globalink award which facilitated his visit to Inria, Paris.
A.K. also acknowledges support from the Fonds de Recherche du Qu\'ebec - Nature et Technologies (FRQNT) via the B2X scholarship for doctoral candidates.
JPT acknowledges the support of the European Union and the French Agence Nationale de la Recherche through the QCDA project.

\bibliographystyle{unsrtabbrev}
\bibliography{references}

\section*{Appendices}
\subsection{The special case of $p=3$}
\label{subsec:p3}
In this appendix, we show that tri-orthogonality applies to $p = 3$ as well.

We first show how addition modulo $9$ can be expressed in ternary.
First we note that any element of $a \in \integer_9$ can be expressed as $a_0 + 3a_1$ for $a_0,a_1 \in \field_3$.

For any two elements $a,b \in \integer_9$, where $a = a_0 + 3a_1$ and $b = b_0 + 3b_1$, we can express addition over $\integer_9$ in terms of addition mod $3$ as
\begin{align}
	\label{eq:base}
	a+b = (a_0 + b_0) - 3(a_0b_0 + a_0^2b_0 + a_0b_0^2) + 3(a_1 + b_1)~,
\end{align}
where the arithmetic within parthentheses is performed modulo $3$.
It is straightforward to check this equation explicitly.
This can easily be extended to the sum over several elements as follows.
\begin{lemma}
	\label{lem:additionMod9}
	For $m\in \natural$, let $\{a^i\}_{i=1}^{m}$ be elements of $\integer_9$, where for each $i$, we may express $a^i$ in ternary as $a^i_0 + 3a^i_1$.
	Then
	\begin{align}
	\label{eq:ind}
		\sum_i a^i = &\left(\sum_i a^i_0\right) - 3\left(\sum_{i<j} a^i_0a^j_0\right) \\
    -&3\left(\sum_{i<j} (a^i_0)^2a^j_0 + a^i_0(a^j_0)^2\right)\nonumber\\
    +&3\left(\sum_{i<j<k}a^i_0a^j_0a^k_0\right) + 3\left(\sum_{i} a^i_1\right)\nonumber~.
	\end{align}
\end{lemma}
\begin{proof}
	The proof follows by induction on $m$.
	The base case for $m=2$ follows from eq. \ref{eq:base}.

	Suppose we have proved eq. \ref{eq:ind} for $m' < m$.
	If we add another element $b = b_0 + 3b_1$, then the non-trivial part emerges from the sum of the $0$-components.
	We have
	\begin{align*}
		\sum_{i=1}^{m'}a^i + b =
		&\left(\sum_{i=1}^{m'}a^i_0 + b_0 \right)\\
    -&3\left(\sum_{i<j} a^i_0 a^j_0 + \sum_{i}a^i_0 b_0\right)\\
    -&3\left(\left(\sum_{i}a^i_0\right)^2b_0 + a^i_0b_0^2\right)\\
		+&\left(\sum_{i<j} (a^i_0)^2a^j_0 + a^i_0(a^j_0)^2\right)\\
    +&3\left(\sum_{i<j<k}a^i_0a^j_0a^k_0\right) + 3\left(\sum_i a^i_1 + b_1\right)~.
	\end{align*}
	Expanding the square of $\sum_i a_i$ yields the desired result.
\end{proof}

\begin{theorem}
  \label{thm:genBraHaah3}
  Let $\h \in \field_p^{m \times n}$ be a tri-orthogonal matrix whose rows are labelled $\{h^a\}_{a=1}^{m}$.
	$\C(\h) = \css(X,\h_0; Z, \g)$ be the tri-orthogonal code obtained from $\h$.
	For $1 \leq a \leq K$, $\epsilon_a := \sum_i (h_i^a) \pmod{9}$.
  The transversal physical gate $U_{2,1}^{\otimes n}$ performs the following transversal logical gate:
	\begin{align*}
		 U_{2,1}^{\otimes n} \encket{u} = \bigotimes_{a=1}^{k} \left(\overline{U_{2,1}}\right)^{\epsilon_a} \encket{u}~.
	\end{align*}
\end{theorem}
\begin{proof}
  Let $f \in \rowspan{h^a}_{a=1}^{m} \subseteq \field_{3}^{n}$ be a vector in the image of $\H$.
  The action of transversal $U_{2,1}$ on $\ket{f}$ can be expressed as
  \begin{align}
    U_{2,1}^{\otimes n}\ket{f} &= \prod_{i=1}^{n}\exp\left(\frac{2\pi i}{9}f_i \right)\ket{f}\\
                   &= \exp\left( \frac{2\pi i}{9} \sum_{i=1}^{n}f_i \right)\ket{f}~.
  \end{align}
  By assumption, we may express $f = \sum_{a} u_{a} h^{a} \pmod{9}$ for some constants $\{u_a\}_a \in \field_3$.
  We express this sum in a ternary as
  \begin{align*}
    f_i &= \sum_{a} u_a h^{a}_{i} \pmod{9}\\
      &= \sum_{a} u_a h^{a}_{i} \pmod{3}\\
             &- 3\sum_{a<b} \big[
                        u_au_b h^{a}_i h^{b}_i + 
                        (u_a)^2 u_b (h^{a}_{i})^2 h^{b}_i + \\
                        &u_a (u_b)^2 h^{a}_{i} (h^{b}_i)^2 \big] \pmod{3}  \nonumber \\
             &- 3\left[\sum_{a<b<c}
             u_a u_b u_c h^{a}_{i}h^{b}_{i} h^{c}_{i} \pmod{3}
                \right]~.
  \end{align*}
  This is especially convenient because each term in the sum above is expressed $\field_3$, the field over which tri-orthogonality is defined.
  This is proved in lemma \ref{lem:additionMod9}.
  We may use definition \ref{def:triorthogonalMatrix} of tri-orthogonality to simplify this expression and obtain
  \begin{align}
    \sum_{i=1}^{n} f_i &= \sum_{a=1}^{k} u_{a} \epsilon_a~.
  \end{align}

  Hence it follows that
  \begin{align}
    U_{2,1}^{\otimes n}\ket{f} &= \bigotimes_{a=1}^{k} \exp\left( \frac{2\pi i}{9} u_{a}\epsilon_a \right)\ket{f}~,
  \end{align}
  which implies the desired result.
\end{proof}

\subsection{$\css(X,\h_0;Z,\g)$ is a valid code}
\label{subsec:validCode}
\validCode
\begin{proof}
  The logical $Z$ operators corresponding to $\h_1$ are going to commute with the $X$ stabilizer generators corresponding to $\h_0$ because of condition \ref{item:biorthog}.
  Note that without loss of generality, we can always replace any row in $h \in \h_1$ with $\alpha h$ for $\alpha \in \integer_p$ such that we modify just the logical $X$ operators.
  By doing so, we can force the $X$ and $Z$ logical operators to obey the canonical commutation relations as $[\logX^{a},\logZ^{b}] = \delta_{a,b}$.

  Following an argument identical to that of lemma 1 of \cite{bravyi2012magic}, we can show that the vectors of $\h_1$ are independent and that they do not overlap with $\h_0$.
  Any vector $f \in \H_1$ can be expressed as $\sum_a x_a h^{a}$ for some $x_a \in \field_p$.
  It follows from the orthogonality condition that if $f = 0$, then each $x_a$ must be identically $0$.
  Similarly if we assume that a $u \in \H_0$ was in $\h_1$, then we run into a contradiction for the same reason.
  Therefore the dimension of the code is $k$.

  We have chosen these matrices such that $\H_0 \subseteq \G$.
  Therefore, the $X$ distance $d_X$ is at least as much as the $Z$ distance $d_Z$ since
  \begin{align}
    d_Z = \min_{u \in \H_1 \setminus \G} (u) \leq \min_{u \in \H_1 \setminus \H_0} (u) = d_X~.
  \end{align}
  This is the distance of the code.
\end{proof}

\subsection{Tri-orthogonality from triply-even codes}
\label{subsec:triplyEvenToTriorthogonal}

\triplyEvenToTriorthogonal*
\begin{proof}
	Since $\tC^{\perp}$ is triply-even, it follows that for any three rows $\th^a, \th^b, \th^c \in \tC^{\perp}$, that $|\th^a * \th^b * \th^c| = 0 \pmod{p}$ by definition.
	If $1 \leq a < b < c \leq m$, then the same is true if we restrict these vectors to the last $n$ locations as they do not overlap on the first $k$ locations.
	This in turn implies that $|h^a * h^b * h^c| = 0 \pmod{p}$ for $h^a, h^b, h^c \in \C^{\perp}$.
	It also implies that for $1 \leq a \leq k$, that $|h^a * h^a * h^a| = 1 \pmod{p}$.
	Since the all $1$s vector is in $\tC^{\perp}$, it follows that if $a < b$ that $|h^a * h^b| = 0 \pmod{p}$ and that for $1 \leq a \leq k$, that $|h^a * h^a| = -1 \pmod{p}$.
\end{proof}

\end{document}